\documentclass[conference]{IEEEtran}
\usepackage{cite}
\usepackage{verbatim}
\usepackage{epsfig,bbm}
\usepackage[colorlinks,bookmarksopen,bookmarksnumbered,citecolor=red,urlcolor=red,]{hyperref}
\usepackage{CJK}
\usepackage{indentfirst}
\usepackage{multirow}
\usepackage{epstopdf}
\usepackage{graphicx}
\usepackage{footmisc}
\graphicspath{{./figures/}}
\usepackage{amsfonts}
\usepackage{mathrsfs}
\usepackage{setspace}
\usepackage{amsmath}
\usepackage{algorithm,algorithmic,amsbsy,amsmath,amssymb,epsfig,bbm,mathrsfs, bbm} 
\usepackage{amsthm}
\usepackage{verbatim}
\usepackage{geometry}
\usepackage[subfigure]{graphfig}

\newtheorem{theorem}{Theorem}

\allowdisplaybreaks

\begin{document}

\title{Analysis of Wireless-Powered Device-to-Device Communications with Ambient Backscattering 
}

 \author{  Xiao Lu$^{\dagger}$, Hai Jiang$^{\dagger}$, Dusit Niyato$^{\ddagger}$, Dong In Kim$^{\star}$, and Ping Wang$^{\ddagger}$ \\ 
    ~$^{\dagger}$ Dept. of Electrical $\&$ Computer Engineering, 
    University of Alberta, Canada\\
	~$^{\ddagger}$ School of Computer Science and Engineering, Nanyang Technological University, Singapore \\
	~$^{\star}$   School of Information and Communication Engineering, Sungkyunkwan University (SKKU), South Korea \\ 
}

\markboth{}{Shell \MakeLowercase{\textit{et al.}}: Bare Demo of
IEEEtran.cls for Journals}
  
\maketitle

\begin{abstract}
Self-sustainable communications based on advanced energy harvesting technologies have been under rapid development, which facilitate autonomous operation and energy-efficient transmission. 
Recently, ambient backscattering that leverages existing RF signal resources in the air has been invented to empower data communication among low-power devices. In this paper, we introduce hybrid device-to-device (D2D) communications by integrating ambient backscattering and wireless-powered communications. 
The hybrid D2D communications are self-sustainable, as no dedicated external power supply is required.
However, since the radio signals for energy harvesting and backscattering come from external RF sources, the performance of the hybrid D2D communications needs to be optimized efficiently. As such, we design two mode selection protocols for the hybrid D2D transmitter, allowing a more flexible adaptation to the environment. We then introduce analytical models to characterize the impacts of the considered environment factors, e.g., distribution, spatial density, and transmission load of the ambient transmitters, on the hybrid D2D communications performance. Extensive simulations show that the repulsion factor among the ambient transmitters has a non-trivial impact on the communication performance. Additionally, we reveal how different mode selection protocols affect the performance metrics.

\end{abstract}
 
\emph{Index terms- Internet-of-Things (IoT), Ambient backscatter,  wireless-powered communications, D2D communications, RF energy harvesting}. 


\section{Introduction}
With the advent of the Internet-of-Things (IoT)~\cite{V.2016Gazis,D.2016Niyato}, intelligent devices, such as smart household devices~\cite{D.2011Niyato},
renewable sensors~\cite{D.ICCNiyato2016},  
vehicular communicators, RFID tags, and wearable health-care gadgets, have become increasingly interconnected at an unprecedented scale.  
In this context, device-to-device (D2D) communications~\cite{Camps-Mur2013}, which empower devices in proximity to establish direct connections without an involvement of any cellular base stations, appear to be a cost-effective and energy-efficient solution. Recent research efforts~\cite{J.2015Liu,X.LuD2D} %
have shown that D2D communications have achieved evident performance gains in terms of network coverage and capacity, peak rates, throughput, communication latency and user experience. Therefore, it is envisioned to be an intrinsic part of the IoT.

Lately, ambient backscatter communications \cite{V.2013Liu} have  appeared to be a promising self-sustainable communication paradigm. In ambient backscattering, the information transmission is done by load modulation which does not involve active RF generation. In particular, an ambient backscattering device tunes the antenna load reflection coefficient by switching between two or more impedances, resulting in a varied amount of incident signal to be backscattered. In principle, when the impedance of the chosen load matches with that of the antenna, a small amount of the signal is reflected, exhibiting a signal absorbing state. Conversely, if the impedances are not matched, a large amount of the signal is reflected, indicating a signal reflecting state. A backscatter transmitter can use an absorbing state or reflecting state to transmit a `0' or `1' bit. Based on the detection of the amount of the reflected signal, the transmitted information is decoded at the receiver.
 
Unlike conventional backscatter communication (e.g., for passive sensors and RFID tags), ambient backscattering functions without the need of a dedicated carrier emitter (e.g., an RFID reader). Instead, an ambient backscatter device utilizes exogenous and incident RF waves as both energy resource to scavenge and signal resource to reflect. Moreover, ambient backscattering is featured with coupled backscattering and energy harvesting processes~\cite{X.2017Lu}. To initiate information transmission, the device first extracts energy from incident RF waves through rectifying. Once the rectified DC voltage is above an operating level of the circuit, the device is activated to conduct load modulation. Simultaneously, backscatter modulation is done on the reflected wave, enabling a full-time transmission. For example, a recent experiment in~\cite{N.2014Parks} demonstrated that a 1 Mbps transmission rate can be achieved at the distance of 7 feet, when the incident RF power available is above -20 dBm.

Despite many benefits, ambient backscatter communications have drawbacks that limit their applicability for D2D communications. Specifically, ambient backscattering achieves  relatively low data rate, typically ranging from several to tens of kbps~\cite{V.2013Liu,B.2014Kellogg}, which largely constrains the applications. A relatively high signal-to-noise ratio (SNR) is required to realize a low-error transmission with modulated backscatter. Moreover, the transmission distance is limited, typically ranging from several feet to tens of feet \cite{V.2013Liu,N.2014Parks} due to severe propagation attenuation and embedded modulation for an intended receiver. To address these shortcomings, in this paper, we introduce a novel hybrid D2D communication paradigm that integrates ambient backscattering and wireless-powered communications~\cite{XLuSurvey,D.2015NiyatoICC,D.2015NiyatoWCNC} as a self-sustainable communication method. For communication, the proposed hybrid D2D transmitter harvests energy from ambient RF signals and can select to perform ambient backscattering or wireless-powered communications with the aim of extending the applicability as well as functionality. Through the analysis, we show that these two technologies can well complement each other and result in better performance for D2D communications.
 
{\bf Notations:} In the following, we use $\mathbb{E}[\cdot]$ to denote the average over all random variables in $[\cdot]$, 
$\mathbb{E}_{X}[\cdot]$ to denote the expectation over the random variable $X$, and $\mathbb{P}[E]$ to denote the probability that an event $E$ occurs. Besides, $\|{\mathbf x}\|$ is used to represent the Euclidean norm between the coordinate ${\mathbf x}$ and the origin of the Euclidean space. $\bar{z}$ and $|z|$ denote the complex conjugate and modulus of the complex number $z$, respectively.  The notations $f_{X}(\cdot)$, $F_{X}(\cdot)$, $M_{X}(\cdot)$ and $\mathcal{L}_{X}(\cdot)$ are used to denote, respectively, the probability density function (PDF), cumulative distribution function (CDF), moment generating function (MGF), and Laplace transform of a random variable $X$.
 
\section{Ambient Backscattering Assisted Wireless-Powered Communications}

We now propose a novel hybrid transmitter that combines two self-sustainable communication approaches, i.e., ambient backscatter communications and wireless-powered communications. On one hand, ambient backscatter communications can be operated with very low power consumption. Thus, ambient backscattering may still be performed when the power density of ambient RF signals is low. On the other hand, the wireless-powered communications, also referred to as harvest-then-transmit (HTT)~\cite{H.2014Ju, X.Sec2016Lu}, though have higher power consumption, can first accumulate harvested energy and achieve possibly longer transmission distance through active RF transmission. Therefore, these two approaches can well complement each other and result in better transmission performance. 

We depict the block diagrams of the hybrid transmitter and hybrid receiver in Fig.~\ref{fig:D2D_system}.  With the proposed architecture, the hybrid transmitter is flexible
to perform active data transmission, backscattering, and RF energy harvesting. At the receiver, a dual-mode circuit  can demodulate data from both the modulated backscatter and active RF transmission. The mode selection can be done by the
hybrid transmitter through signaling.

We consider the hybrid D2D communications coexisting with ambient RF transmitters, e.g., cellular base stations and mobiles. Fig.~\ref{fig:D2D_system} illustrates our proposed system model.  
We consider two groups of coexisting ambient transmitters, denoted as $\Phi$ and $\Psi$, respectively, which work on different frequency bands. The RF energy harvester of the hybrid transmitter denoted as $\mathrm{S}$ scavenges on the transmission frequency of $\Phi$ (e.g., ambient base stations). If the hybrid transmitter is in ambient backscattering mode, it performs load modulation on the incident signals from $\Phi$. Alternatively, when the hybrid transmitter is in HTT mode, it harvests energy from ambient transmitters in $\Phi$, and transmits over a different frequency band used by ambient transmitters in $\Psi$ (e.g., ambient mobile users). The received signal at the hybrid receiver denoted as $\mathrm{D}$ from the hybrid transmitter is impaired by the interference from $\Psi$. We assume $\Phi$ and $\Psi$ follow independent $\alpha$-Ginibre point process (GPP)~\cite{DecreusefondFlintVergne}, where $\alpha \in \big(0, 1\big]$ represents the repulsion factor which measures the correlation among the spatial points in $\Phi$ and $\Psi$, respectively. 

\begin{figure}
	\centering
	\includegraphics[width=0.4\textwidth]{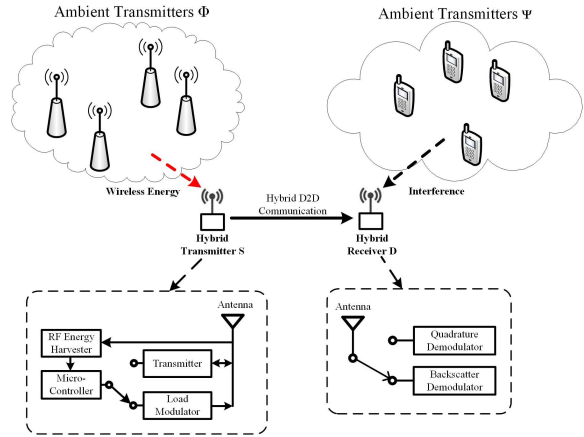}  
	\caption{Illustration of the hybrid D2D communication.} \label{fig:D2D_system}
\end{figure}  
 
Without loss of generality,  
the point processes $\Phi$ and $\Psi$ are assumed to be supported on the circular observation windows $\mathbb{O}_{\mathrm{S}}$ and  $\mathbb{O}_{\mathrm{D}}$ with radius $R$, which are centered at  $\mathrm{S}$ and $\mathrm{D}$, respectively. 
The transmit power of the ambient transmitters belonging to $\Phi$ and $\Psi$ are denoted as $P_{A}$ and $P_{B}$, respectively. Let $\zeta_{A}$ and $\zeta_{B}$ denote the spatial density of $\Phi$ and $\Psi$, respectively.  
Let $\mathcal{A}$ (and $\mathcal{B}$) denote the set of active ambient transmitters of $\Phi$ (and $\Psi$) observed in $\mathbb{O}_{\mathrm{S}}$ (and $\mathbb{O}_{\mathrm{D}}$).  
We assume that the probability that an ambient transmitter in  $\Phi$ (and $\Psi$) is active is an independent and identically distributed (i.i.d.) random variable, denoted as $l_{A}$ (and  $l_{B}$). $l_{A}$ and $l_{B}$ can also be interpreted as the {\it transmission load}, which measures the fraction of time that an ambient transmitter in $\Phi$ and $\Psi$, respectively, is active. It is worth noting that the sets of active transmitters in $\Phi$ and $\Psi$ in the reference time are independent thinning point processes of $\Phi$ and $\Psi$ with spatial density $l_{A} \zeta_{A}$ and $l_{B} \zeta_{B}$, respectively. 
Let $\xi$ represent the ratio between $\zeta_{B}$ and $\zeta_{A}$, i.e., $\xi=\zeta_{B}/\zeta_{A}$, referred to as the interference ratio. A larger value of $\xi$ indicates a higher level of interference.

Let $\mathbf{x}_{\mathrm{S}}$ represent the location of the hybrid transmitter.
The power of the incident RF signals at the antenna of $\mathrm{S}$ can be calculated as
$P_{I}= P_{A} \sum_{a \in \mathcal{A}} h_{a,\mathrm{S}} \|\mathbf{x}_{a}-\mathbf{x}_{\mathrm{S}}\|^{-\mu}$, 
where 
$h_{a,\mathrm{S}}$ represents the channel gain between the ambient transmitter $a \in \mathcal{A}$ and $\mathrm{S}$, and $\mu$ denotes the path loss exponent.
The circuit of the hybrid transmitter becomes functional if it can extract sufficient energy from the incident RF signals. When the hybrid transmitter works in different modes (i.e., either HTT or ambient backscattering),  the hardware circuit consumes different amounts of energy.\footnote{The typical power consumption rate of an RF-powered transmitter ranges from hundreds of micro-Watts to a few milli-Watts \cite{X.CRNLu2014,X.2015Lu,Y.2015Ishikawa}, %
while that of a backscatter transmitter ranges from a few micro-Watts to hundreds of micro-Watts~\cite{N.2014Parks}.} Let $\rho_{\mathrm{B}}$ and $\rho_{\mathrm{H}}$ denote the circuit power consumption rates (in Watt) in ambient backscattering and HTT modes, respectively. If the hybrid transmitter cannot harvest sufficient energy, an outage occurs.

In the ambient backscattering mode, if the instantaneous energy harvesting rate (in Watt) exceeds $\rho_{\mathrm{B}}$, the hybrid transmitter can perform backscatter modulation. 
During backscattering process, a fraction of the incident signal power, denoted as $P_{H}$, is rectified for conversion from RF signal to direct current (DC), and the residual amount of signal power, denoted as $P_{R}$, is reflected to carry the modulated information. 
In the ambient backscattering mode, the energy harvesting rate (in Watt) can be represented as~\cite{C.2012Boyer,C.2014Boyer}
$P^{\mathrm{B}}_{E}=\beta P_{H}=\beta \eta P_{I}$,
where $0 < \beta \leq 1$ denotes the efficiency of RF-to-DC energy conversion, and $\eta$ represents the fraction of the incident RF power for RF-to-DC energy conversion. Note that the value of $\eta$ depends on the symbol constellation adopted for multi-level load modulation~\cite{C.2012Boyer}. For example, $\eta$ is 0.625 on average assuming equiprobable symbols if binary constellations are adopted with modulator impedance values set at 0.5 and 0.75~\cite{C.2014Boyer}.  
 
Let $\mathbf{x}_{\mathrm{D}}$ represent the location of the hybrid receiver.
$d\!\!=\!\!\|\mathbf{x}_{\mathrm{S}}\!-\! \mathbf{x}_{\mathrm{D}}\|$ denotes the distance between $\mathrm{S}$ and $\mathrm{D}$. 
Then, in ambient backscattering mode, the power of the received backscatter at $\mathrm{D}$ from $\mathrm{S}$
can be calculated as $P_{\mathrm{S},\mathrm{D}} \!=\!
    \delta P_{I}  (1-\eta) h_{\mathrm{S},\mathrm{D}} d^{-\mu}$ if $P^{\mathrm{B}}_{E} > \rho_{\mathrm{B}}$ and $P_{\mathrm{S},\mathrm{D}} = 0$ otherwise,   
where $0<\delta \leq 1$ is the backscattering efficiency of the transmit antenna, which is related to the antenna aperture~\cite{V.2006Nikitin}, and $h_{\mathrm{S},\mathrm{D}}$ denotes the channel gain between $\mathrm{S}$ and $\mathrm{D}$ on the transmission frequency of $\Phi$.   
If $\mathrm{S}$ is active in the ambient backscattering mode, the resulted SNR 
at $\mathrm{D}$ is
\begin{equation} \label{eqn:SNR_B}
	\nu_{\mathrm{B}} =\frac{P_{\mathrm{S},\mathrm{D}}  }{ \sigma^2}=\frac{\delta P_{I} (1-\eta) h_{\mathrm{S},\mathrm{D}}}{ d^{\mu}\sigma^2},
\end{equation}
where $\sigma^2$ is the power spectrum density of additive white Gaussian noise (AWGN).
If the received SNR $\nu_{\mathrm{B}}$ is above a threshold $\tau_{\mathrm{B}}$, $\mathrm{D}$ is able to successfully decode information from the modulated backscatter at a pre-designed rate $T_{\mathrm{B}}$ (in bits per second (bps)). This backscatter transmission rate is dependent on the setting of resistor-capacitor circuit elements~\cite{V.2013Liu}. 

When the hybrid transmitter $\mathrm{S}$ chooses to adopt active RF transmission, it is operated by the HTT protocol~\cite{H.2014Ju}. In the HTT mode, the hybrid transmitter works in a time-slot based manner. Specifically, in each time slot, the first period, with time fraction $\omega$, is for harvesting energy, during which the impedance of the matching network is tuned to fully match to that of the antenna to maximize the energy conversion efficiency. The corresponding energy harvesting rate is $P^{\mathrm{H}}_{E}=\omega \beta P_{I}$.
This harvested energy is first utilized to power the circuit. Then the remaining energy, if available, is stored in an energy storage. 
If the harvested energy is enough to operate the circuit, the hybrid transmitter spends the rest of the period $(1-\omega)$ to perform active transmission with the stored energy.  
     
In the active transmission phase, the transmit power of $\mathrm{S}$ is $P_{\mathrm{S}} \! =    \! \frac{ P^{\mathrm{H}}_{E}- \rho_{\mathrm{H}}}{1-\omega }$ if $P^{\mathrm{H}}_{E}   \! >   \! \rho_{\mathrm{H}}$ and $P_{\mathrm{S}}  \!  =  \! 0$ otherwise.
Then, the received signal-to-interference-and-noise ratio (SINR) at $\mathrm{D}$ can be expressed as
\begin{eqnarray}
\nu_{\mathrm{H}} =\frac{ P_{\mathrm{S}} \widetilde{h}_{\mathrm{S},\mathrm{D}} d^{-\mu} }{\sum_{b \in \mathcal{B}}  P_{B} h_{b,\mathrm{D}}\|\mathbf{x}_{b}-\mathbf{x}_{\mathrm{D}}\|^{-\mu} +\sigma^{2} },
\end{eqnarray}
where $\widetilde{h}_{\mathrm{S},\mathrm{D}}$ denotes the channel gain between $\mathrm{S}$ and $\mathrm{D}$ on the transmit frequency of $\Psi$, and $h_{b,\mathrm{D}}$ is the channel gain from ambient transmitter $b \in \mathcal{B}$ to $\mathrm{D}$. 

As the hybrid D2D communications and the transmission from ambient transmitters may be performed in different environments,
we consider different fading channels for $h_{\mathrm{S},\mathrm{D}}$, $\widetilde{h}_{\mathrm{S},\mathrm{D}}$, $h_{a,\mathrm{S}}$ and $h_{b,\mathrm{D}}$. Specifically,  $h_{\mathrm{S},\mathrm{D}}$ and $\widetilde{h}_{\mathrm{S},\mathrm{D}}$ are assumed to follow Rayleigh distribution. Both $h_{a,\mathrm{S}}$ and $h_{b,\mathrm{D}}$ follow i.i.d. Nakagami-$m$ distribution, which is a general channel fading model that contains Rayleigh distribution as a special case when $m=1$.  
This channel model allows a flexible evaluation of the impact of the ambient signals.    
Let $\mathcal{G}(x,y)$ represent the gamma distribution with shape parameter $x$ and scale parameter $y$, and $\mathcal{E}(x)$ represent the exponential distribution with rate parameter $x$. Thus, the channel gain coefficients are expressed as $h_{a,\mathrm{S}}, h_{b,\mathrm{D}} \sim \mathcal{G}( m, 1/m)$ and $h_{\mathrm{S},\mathrm{D}}, \widetilde{h}_{\mathrm{S},\mathrm{D}} \sim \mathcal{E}(1)$.  
 
For the operation of our proposed hybrid transmitter, we consider two mode selection protocols, namely, \emph{power threshold-based protocol} (PTP) and \emph{SNR threshold-based protocol} (STP). 
\begin{itemize}

\item With PTP, a hybrid transmitter first detects the available energy harvesting rate $P^{\mathrm{H}}_{\mathrm{E}}$. If $P^{\mathrm{H}}_{\mathrm{E}}$ is below the threshold which is needed to power the RF transmitter circuit (for active transmission), i.e., $ P^{\mathrm{H}}_{E} \leq \rho_{\mathrm{H}}  $, the ambient backscattering mode will be used. Otherwise, the HTT mode will be adopted. 

\item With STP, the hybrid transmitter first attempts to transmit by backscattering. If the achieved SNR at the receiver is above the threshold which is needed to decode information from the backscatter, i.e., $\nu_{\mathrm{B}} > \tau_{\mathrm{B}}$, the transmitter will be in the ambient backscattering mode. Otherwise, it will switch to the HTT mode.

\end{itemize} 
 
 
\section{Performance Analysis}

In this section, we focus on analyzing the {\em coverage probability} of the hybrid D2D transmitter which is defined as the probability that the hybrid transmitter can successfully transmit data to its receiver. 
The transmission of the hybrid transmitter is successful if the achieved SNR or SINR at the associated receiver exceeds its corresponding target threshold $\tau_{\mathrm{B}}$ or $\tau_{\mathrm{H}}$ in backscattering mode or in HTT mode, respectively.  
Let $\mathcal{C}_{\mathrm{M}}$ denote the coverage probability of the hybrid transmitter being in mode $\mathrm{M} \in \{\mathrm{B},\mathrm{H}\}$. Then, the overall coverage probability is defined as 
\begin{eqnarray} \label{def:coverage_probability} 
\mathcal{C} \triangleq  \mathbb{E} [ \mathbb{E}_{\mathrm{M}}[\mathcal{C}_{\mathrm{M}}|\mathrm{M}]] = \mathbb{E} \big[ \mathbb{E}_{\mathrm{M}}[\mathbb{P}[\nu_{\mathrm{M}}> \tau_{\mathrm{M}}, P^{\mathrm{M}}_{E} > \rho_{\mathrm{M}}|\mathrm{M}]]\big]	.	
\end{eqnarray}  

In the following, the performance analysis of the hybrid D2D communication is based on a general class of stochastic geometry tool, namely $\alpha$-GPP~\cite{DecreusefondFlintVergne}.
$\alpha$-GPP is a repulsive point process which allows to characterize the repulsion among the distribution of the randomly located points and has the Poisson point process (PPP) as a special case (i.e., when $\alpha \to 0$). Recently, $\alpha$-GPP has attracted considerable attention in wireless network modeling \cite{X.2016Lu,I.2005Flint,X.Lu2015,I.2014Flint} 
because it renders tractable analytical expressions in terms of Fredholm determinants. The Fredholm determinant is a generalized determinant of a matrix defined by bounded operators on a Hilbert space and has shown to be an efficient way for numerical evaluation of the relevant quantities~\cite{L.2015Decreusefond}.
We refer to~\cite{ DecreusefondFlintVergne,L.2015Decreusefond} for a formal definition of an $\alpha$-GPP and the mathematical details.

Based on the $\alpha$-GPP framework, we have the coverage probability of PTP described as follows.

\begin{theorem}
\label{thm:CoverageOutage_PTP} 
The coverage probability of the hybrid D2D communications under PTP is
\begin{align}
\label{eq:coverage_probability_Rayleigh_PTP}
  \mathcal{C}_{\mathrm{PTP}} \!& =  \!  F_{P_I}\Big(\frac{\rho_{\mathrm{H}}}{\omega \beta}\Big)
      \int^{\infty}_{\rho_{\mathrm{B}}/\beta \eta} \exp \left( - \frac{  \tau_{\mathrm{B}}d^{\mu} \sigma^2 }{\delta \rho \left (1- \eta \right ) } \right) f_{P_{I}}(\rho) \mathrm{d}\rho \nonumber \\ &+
 \left(\!1\!- \! F_{P_I}\Big(\frac{\rho_{\mathrm{H}}}{\omega \beta}\Big) 
 \right) \int^{\infty}_{\rho_{\mathrm{H}}/\beta \omega}  \exp \left(\! - \frac{  \tau_{\mathrm{H}}d^{\mu} \sigma^{2} (1-\omega)}{\omega \beta \rho - \rho_{\mathrm{H}}} \!\right) \nonumber \\
 & \times \mathrm{Det}\left( \mathrm{Id} + \alpha \mathbb{B}_{\Psi}(\rho) \right)^{-1/\alpha} f_{P_{I}}(\rho) \mathrm{d}\rho   
\end{align} 
where $\mathrm{Det}(\cdot)$ represents the Fredholm determinant~\cite{L.2015Decreusefond}, $F_{P_{I}}(\rho)$ and $f_{P_{I}}(\rho)$ are the CDF and PDF of $P_{I}$ given, respectively, as 
\begin{equation} \label{CDF_PI}
F_{P_{I}}(\rho)=\mathcal{L}^{-1}\left\{ \frac{\mathrm{Det}(\mathrm{Id}+\alpha \mathbb{A}_{\Phi}(s ))^{-1/\alpha }}{s} \right\}(\rho),
\end{equation}
and 
\begin{equation}
f_{P_{I}}(\rho) =\mathcal{L}^{-1}\{ \mathrm{Det}(\mathrm{Id}+\alpha \mathbb{A}_{\Phi}(s))^{-1/\alpha }\}(\rho), \label{eq:PDF}
\end{equation} 
wherein $\mathcal{L}^{-1}$ means inverse Laplace transform and $\mathbb{A}_{\Phi}(s)$ is 
given by
\begin{align}
\label{eqn:kernal_A}
\mathbb{A}_{\Phi}(s)&=\sqrt{1-\left( 1+ \frac{s  P_{A}}{m\|\mathbf{x}-\mathbf{x}_{\mathrm{S}}\|^{\mu}} \right)^{-m}} G_{\Phi}(\mathbf{x},\mathbf{y})  \nonumber \\  & \hspace{20mm} \times\sqrt{1-\left( 1+ \frac{s  P_{A}}{m\|\mathbf{y}-\mathbf{x}_{\mathrm{S}}\|^{\mu}} \right)^{-m}}, 
\end{align}  
and $\mathbb{B}_{\Psi}(\rho)$ is
\begin{align} \label{eqn:kernel_B}
& \mathbb{B}_{\Psi} (\rho) \! = \! \!  \sqrt{\!  1 \!  - \!  \left(\!\!  1\!  +\!  \frac{   \tau_{\mathrm{H}}d^{\mu} (1\!-\!\omega)  P_{B} }{m(\omega \beta \rho\! - \! \rho_{\mathrm{H}})\|\mathbf{x}\!-\!\mathbf{x}_{\mathrm{D} }\|^{\mu} } \!  \right)^{\!\!\! -m} }\!\! G_{\Psi} (\mathbf{x},\mathbf{y}) \nonumber \\
& \hspace{20mm} \times \sqrt{\! 1\!-\! \left(\! \! 1\! +\! \frac{    \tau_{\mathrm{H}}d^{\mu} (1\!-\!\omega)   P_{B} }{m(\omega \beta \rho\! - \!\rho_{\mathrm{H}})\|\mathbf{y}\!-\!\mathbf{x}_{\mathrm{D}}\|^{\mu} }\!\! \right)^{\!\!\!-m} },  \hspace{-2mm}
\end{align}
wherein $G_{\Psi}$ is the Ginibre kernel of $\Psi$ defined as
\begin{equation}
\label{eq:ginibre_Psi} G_{\Psi}(\mathbf x,\mathbf y) \! = \! l_{B} \zeta_{B} \,e^{\pi l_{B} \zeta_{B} \mathbf x \bar{\mathbf y}} e^{-\frac{\pi l_{B} \zeta_{B}}{2}( |\mathbf x|^2 + |\mathbf y|^2)}, 
\mathbf x,\mathbf y \in \mathcal{B}.	
\end{equation}

\end{theorem} 
  
\begin{proof} 
We first determine the distribution of the aggregated received power from ambient transmitters at the origin by  calculating its Laplace transform. Specifically, we have 
\begin{align} 
\mathcal{L}_{P_{I}}(s)  
&= \mathbb{E} \left[ \exp \left ( \sum_{a \in \mathcal{A}} \ln \left (\left( 1+ \frac{s  P_{A} }{m\|\mathbf{x}-\mathbf{x}_{\mathrm{S}}\|^{\mu}} \right)^{-m} \right) \right) \right] 
	\nonumber \\ 	
&\overset{\text{(i)}}{=} \mathrm{Det}(\mathrm{Id}+\alpha \mathbb{A}_{\Phi}(s))^{-1/\alpha }, \label{LP_incident_rate}
\end{align}	
	where $M_h(\cdot)$ is the MGF of $h_{a,\mathrm{S}}$ and $(\mathrm{i})$ follows by applying {\em \cite[Theorem 2.3]{L.2015Decreusefond}}, and $\mathbb{A}_{\Phi}$  is given in (\ref{eqn:kernal_A}).  
	
	Given the Laplace transforms of $P_{I}$, by definition, the PDF of $P_{I}$ is attained by taking the inverse Laplace transform as follows
	\begin{align}
	f_{\!P_{I}}\!(\rho)&   =\mathcal{L}^{ -1} \{ \mathrm{Det}(\mathrm{Id}+\alpha \mathbb{A}_{\Phi}(s))^{\!-1/\alpha }\!\}(\rho), \label{eq:PDF_P_I}
	\end{align}
	
	Furthermore, integrating PDF in (\ref{eq:PDF_P_I}) yields
	\begin{align} \label{eq:CDF_PI}
	F_{P_{I}}\!(\rho)  
	= \mathcal{L}^{-1} \!\left \{ \! \frac{ \mathrm{Det}(\mathrm{Id}+\alpha \mathbb{A}_{\Phi}(s))^{-\frac{1}{\alpha} }}{s} \!\!\right \}(\rho)	.
	\end{align}

	One notices that the probability that $\mathrm{S}$ is in ambient backscattering mode under PTP, denoted as $\mathcal{B}_{\mathrm{PTP}}$, is equal to the CDF of  $P_{I}$ evaluated at $\frac{\rho_{\mathrm{H}}}{\omega \beta}$, which is expressed as 
		\begin{align}\label{eqn:B_PTP}
	\hspace{-2mm}	\mathcal{B}_{\mathrm{PTP}}\! = \!  F_{P_{I}} \!\! \left( \frac{\rho_{\mathrm{H}}}{\omega \beta }\right) \!=\! \mathcal{L}^{ -1}\!\! \left\{\! \frac{ \mathrm{Det}(\mathrm{Id}\!+\!\alpha \mathbb{A}_{\Phi}(s  )^{-1/\alpha } }{s } \! \right\}\!\!\left(\!\frac{\rho_{\mathrm{H}}}{\omega \beta}\!\right)\!\!	.\hspace{-2mm}
		\end{align}

We then continue to calculate the coverage probability in the ambient backscattering mode. By the definition of $\mathcal{C}_{\mathrm{B}} $, we have
\begin{align}
\mathcal{C}_{\mathrm{B}}  
& = \mathbb E_{P_{I}}\!\! \left[ \mathbb{P} \! \left[h_{\mathrm{S},\mathrm{D}}\! >\!\frac{ \tau_{\mathrm{B}}d^{\mu} \sigma^2 }{\delta P_{I} \left ( 1- \eta \right ) } \Bigg|  P_{I} \right]\!\! \mathbbm{1}_{P_{I} >\rho_{\mathrm{B}}/\beta \eta} \right] \nonumber\\ 
& 
= \int^{\infty}_{\rho_{\mathrm{B}}/\beta\eta} \exp \left( - \frac{ \tau_{\mathrm{B}} d^{\mu} \sigma^2 }{\delta \rho \left ( 1- \eta \right ) } \right) f_{P_{I}}(\rho) \mathrm{d}\rho 	.	\label{eqn:delta_B_PTP}
\end{align} 
where $\mathbbm{1}_{E}$ denotes the indicator function of an event $E$ which is equal to 1 when $E$
holds and 0 otherwise.

Let $Q=\xi \sum_{b \in \mathcal{B}} P_{B} h_{b,\mathrm{D}} \|\mathbf{x}_{b}-\mathbf{x}_{\mathrm{D}}\|^{-\mu}$ denote the aggregated interference at the receiver. From (\ref{def:coverage_probability}), we then derive the coverage probability in the HTT mode as in (\ref{eqn:coverageprobability_RSP_HTT}) after some mathematical manipulations.
	\begin{align}
	\mathcal{C}_{\mathrm{H}}  &= \int^{\infty}_{\rho_{\mathrm{H}}/\beta\omega } \exp \left(\! - \frac{  \tau_{\mathrm{H}}d^{\mu} \sigma^{2} (1-\omega) }{\omega \beta \rho - \rho_{\mathrm{H}} }\! \right) \nonumber \\ 
	& \hspace{20mm} \times \mathrm{Det}\left( \mathrm{Id} + \alpha \mathbb{B}_{\Psi} (\rho) \right)^{-1/\alpha} f_{P_{I}}(\rho) \mathrm{d}\rho, 
	 \label{eqn:coverageprobability_RSP_HTT} 
	\end{align}
	where  
	 $\mathbb{B}_{\Psi}(\rho)$ is defined in 
	(\ref{eqn:kernel_B}).

	By definition in (\ref{def:coverage_probability}), the coverage probability under PTP can be written as   
	\begin{align}
	\mathcal{C}_{\mathrm{PTP}} = \mathcal{B}_{\mathrm{PTP}} \mathcal{C}_{\mathrm{B}}+ (1-\mathcal{B}_{\mathrm{PTP}})\mathcal{C}_{\mathrm{H}} 	.	\label{eqn:overall_coverage_probability}
	\end{align}

	Then, by plugging $\mathcal{B}_{\mathrm{PTP}}$ in (\ref{eqn:B_PTP}), $\mathcal{C}_{\mathrm{B}}$ in (\ref{eqn:delta_B_PTP}) and $\mathcal{C}_{\mathrm{H}} $ in (\ref{eqn:coverageprobability_RSP_HTT}) into (\ref{eqn:overall_coverage_probability}), we have  
	(\ref{eq:coverage_probability_Rayleigh_PTP}).
\end{proof}

Moreover, we derive the coverage probability for STP in the following Theorem. 
\begin{theorem}
\label{thm:CoverageOutage_STP} 
The coverage probability of the hybrid D2D communications under STP is given as 
\begin{align}
\label{eq:coverage_probability_STP}
\mathcal{C}_{\mathrm{STP}}  & \!   =  \! \!   \int^{\infty}_{\frac{\rho_{\mathrm{H}}}{\beta \omega} } \!      \exp  \! \left(   - \frac{ \tau_{\mathrm{H}} d^{\mu} \sigma^{2} (1\!-\!\omega) }{\omega \beta \rho \!-\! \rho_{\mathrm{H}} } \! \right)  \mathrm{Det} \!\left( \mathrm{Id} \! + \! \alpha \mathbb{B}_{\Psi} (\rho) \right)^{   -1/\alpha}   \nonumber \\
&\times  f_{P_{I}}(\rho) \mathrm{d}\rho 
 \times 
\left [ \int^{\infty}_{\frac{\rho_{\mathrm{B}}}{\beta \eta}} \exp \left( - \frac{  \tau_{\mathrm{B}}d^{\mu} \sigma^2 }{\delta \rho(1 - \eta)} \right) f_{P_{I}}(\rho) \mathrm{d}\rho \right]^2 \nonumber \\
 &  + \int^{\frac{\rho_{\mathrm{B}}}{\beta \eta}}_{0} \exp \left( - \frac{ \tau_{\mathrm{B
}}d^{\mu} \sigma^2 }{\delta \rho(1 - \eta)} \right) f_{P_{I}}(\rho) \mathrm{d}\rho, 
\end{align} 
where  
$f_{P_{I}}(\rho)$ has been obtained in (\ref{eq:PDF}), 
and $\mathbb{B}_{\Psi} (\rho)$ have been defined in  (\ref{eqn:kernel_B}).

\end{theorem}

\begin{proof} 
Let $\mathcal{B}_{\mathrm{STP}}$ denote the probability that $\mathrm{S}$ is in ambient backscattering mode under STP.
According to the criteria of STP, $\mathcal{C}_{\mathrm{STP}}$ can be expressed by $\mathcal{C}_{\mathrm{PTP}}$ in (\ref{eqn:overall_coverage_probability}) with $\mathcal{B}_{\mathrm{PTP}}$ replaced by $\mathcal{B}_{\mathrm{STP}}$.  
One simply notes that the definition of $\mathcal{B}_{\mathrm{STP}} $ is equivalent to  the expression of $\mathcal{C}_{\mathrm{B}}$ in (\ref{eqn:delta_B_PTP}). 
Hence, we have
\begin{align}
\mathcal{B}_{\mathrm{STP}}  
= \int^{\infty}_{\rho_{\mathrm{B}}/\beta\eta} \exp \left( - \frac{ \tau_{\mathrm{B}} d^{\mu} \sigma^2 }{\delta \rho \left ( 1- \eta \right ) } \right) f_{P_{I}}(\rho) \mathrm{d}\rho 	.	\label{eqn:delta_B_STP} 
\end{align}

Therefore, (\ref{eq:coverage_probability_STP}) can be obtained from (\ref{eq:coverage_probability_Rayleigh_PTP}) through the aforementioned replacement. 
\end{proof}

\section{Performance Evaluation and Analysis}

In this section, we validate our derived analytical expressions and conduct performance analysis based on numerical simulations.  
The transmit power levels of the transmitters in $\Phi$ and $\Psi$ are set to be $P_{A}=P_{B}=0.2$ W, which are within the typical range of uplink transmit power for mobile devices.  
When the hybrid transmitter is in the HTT mode, we assume equal time duration for energy harvesting and information transmission. In the ambient backscattering mode, 
the hybrid transmitter functions with power consumption $\rho_{\mathrm{B}} =8.9$ $\mu$W and achieves a data rate of $T_{\mathrm{B}}=$1 kbps if the transmission is successful.
The other system parameters adopted in this section are listed in Table~\ref{parameter_setting} unless otherwise stated.

\begin{table*} 
 \centering
 \caption{\footnotesize Parameter Setting.} \label{parameter_setting} 
 \begin{tabular}{|l|l|l|l|l|l|l|l|l|l|l|l|} 
 \hline
 Symbol & $\mu$ & $d$ & $R$ &   $\eta$ & $\beta$ & $\delta$ & $\tau_{\mathrm{H}}$ & $\tau_{\mathrm{B}}$ & $\rho_{\mathrm{H}}$ & $\sigma^2$ \\ 
 \hline
 Value & 4 & 5 m & 100 m &   0.625 & 30 $\%$ & 1 & -40 dB & 5 dB & 113 $\mu$W & -90 dBm \\
 \hline 
 \end{tabular}
 \end{table*}

 \begin{figure} 
 \centering
  \begin{minipage}[c]{0.4\textwidth}
  \includegraphics[width=0.999\textwidth]{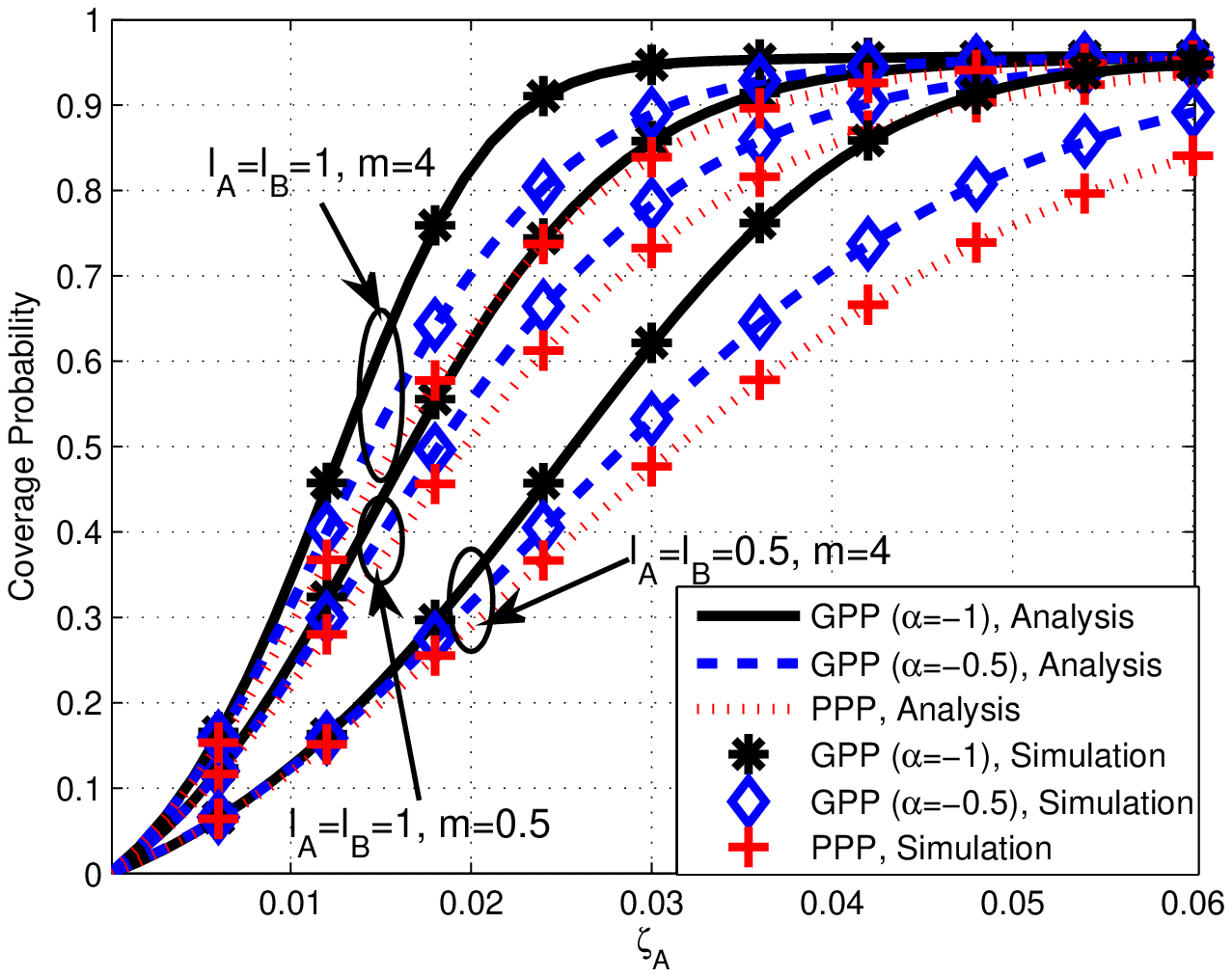} \vspace{-5mm}
  \caption{$\mathcal{C}_{\mathrm{PTP}}$ as a function of $\zeta_{A}$. } \label{fig:CoverProb_density_PTP} 
  \end{minipage}  
  \begin{minipage}[c]{0.4\textwidth}
  \includegraphics[width=0.999\textwidth]{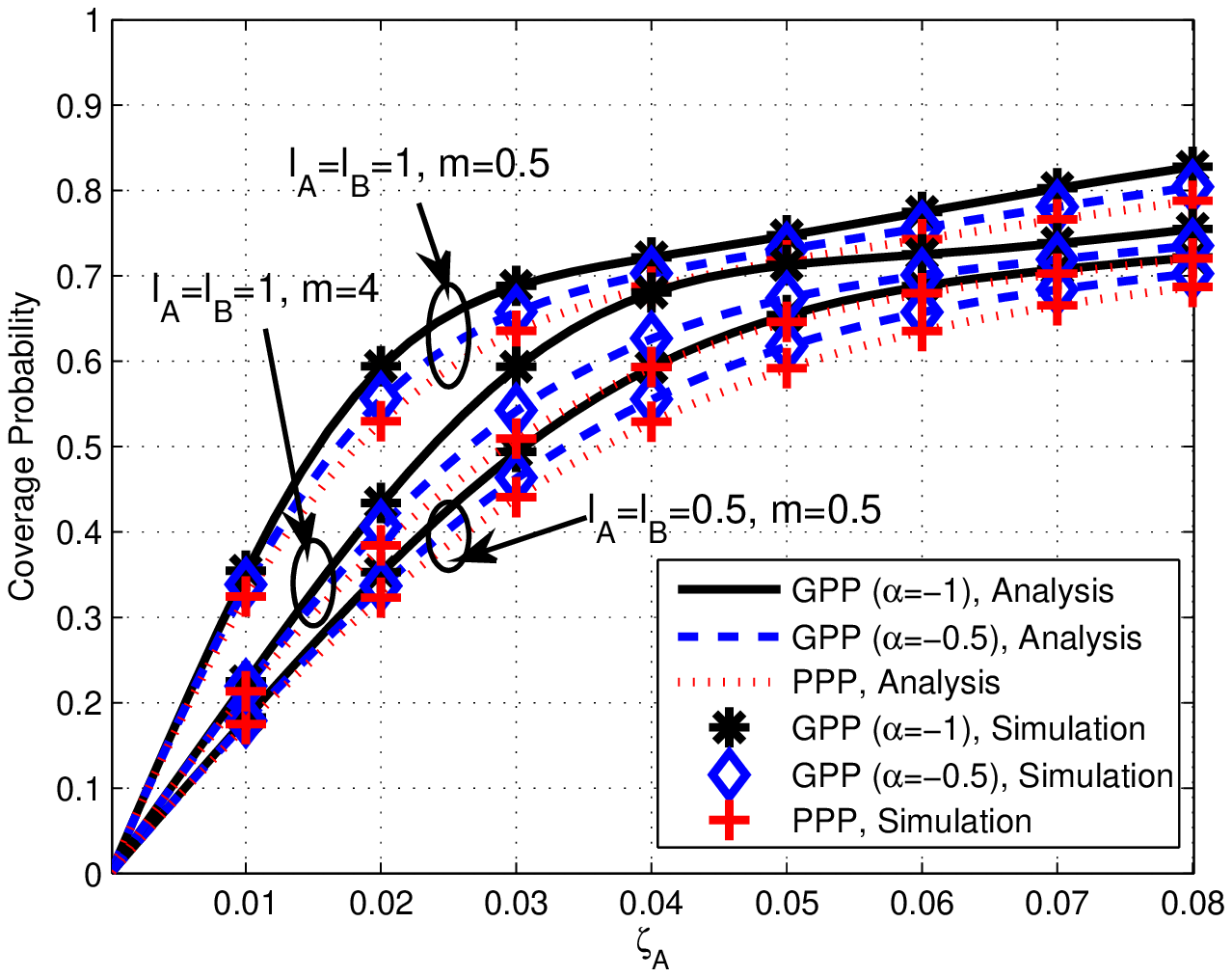} \vspace{-5mm}
  \caption{$\mathcal{C}_{\mathrm{STP}}$ as a function of $\zeta_{A}$. }\label{fig:CoverProb_density_STP} 
  \end{minipage}
  \end{figure}  
For comparison purpose, we evaluate the performance of a pure wireless-powered transmitter operated by the HTT protocol and a pure ambient backscatter transmitter as references, the plots of which are labeled as  ``Pure HTT" and ``Pure Ambient Backscattering", respectively. The performance of a pure wireless-powered transmitter (called pure HTT transmitter) and a pure ambient backscatter transmitter can be obtained by setting the hybrid transmitter in HTT mode and ambient backscattering mode, respectively, in all conditions. Specifically, the  coverage probabilities 
of the pure ambient backscatter transmitter and the pure HTT transmitter can be evaluated by  $\mathcal{C}_{\mathrm{B}}$ in (\ref{eqn:delta_B_PTP})  and
$\mathcal{C}_{\mathrm{H}}$ in (\ref{eqn:coverageprobability_RSP_HTT}), respectively.

\begin{figure} 
 \centering
  \subfigure []
   {
 \label{fig:CP_density_comparison1}
  \centering 
  \includegraphics[width=0.4 \textwidth]{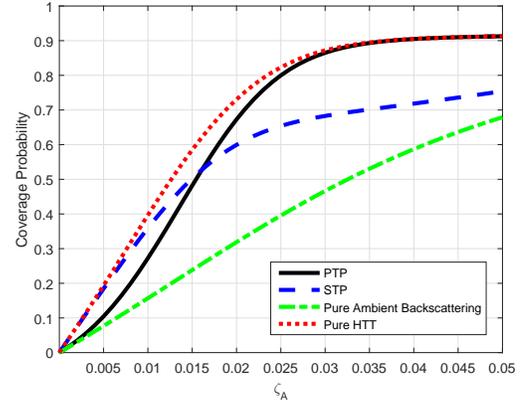}}  
  \centering    
  \subfigure [ 
  ] {
 \label{fig:CP_density_comparison2}
  \centering
 \includegraphics[width=0.4 \textwidth]{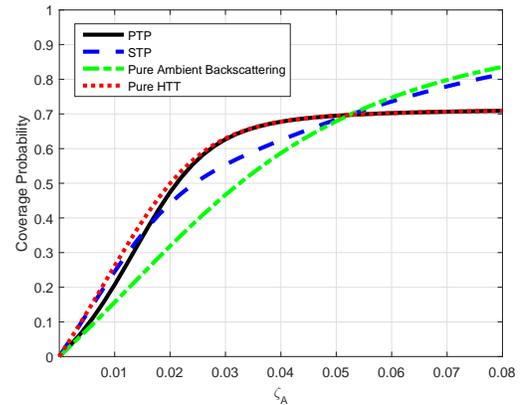}}
 \caption{Comparison of coverage probabilities as a function of $\zeta_{A}$. ((a) $\xi=0.2$, (b) $\xi=0.8$) } 
 \centering
  \vspace{-5mm}
 \label{CP_density_comparison}
 \end{figure}

Fig.~\ref{fig:CoverProb_density_PTP} and Fig.~\ref{fig:CoverProb_density_STP} illustrate how the coverage probabilities $\mathcal{C}_{\mathrm{PTP}}$ and $\mathcal{C}_{\mathrm{STP}}$ obtained in (\ref{eq:coverage_probability_Rayleigh_PTP}) and  (\ref{eq:coverage_probability_STP}), respectively, vary with ambient transmitter density $\zeta_{A}$.
The accuracy of the coverage probability expressions is validated by the simulation results with different values of $\alpha$ under different transmission loads and Nakagami fading coefficients. 
In principle, the improvement of the coverage probability can be achieved with increased transmit power at the hybrid transmitter (either in ambient backscattering mode or HTT mode), thus requiring more incident power. Correspondingly, the incident power becomes higher with larger density $\zeta_{A}$, repulsion factor $\alpha$, transmission load~$l_{A}$, and Nakagami shape parameter $m$. 
The mentioned tendencies of the coverage probability have been verified for both PTP and STP in Fig.~\ref{fig:CoverProb_density_PTP} and Fig.~\ref{fig:CoverProb_density_STP}, respectively, which indicates that both $\mathcal{C}_{\mathrm{PTP}}$ and $\mathcal{C}_{\mathrm{STP}}$ are monotonically increasing functions of $\zeta_{A}$, $\alpha$, $l_{A}$ and $m$. Note that from Fig.~\ref{fig:CoverProb_density_PTP}, with the increase of $\zeta_{A}$, the coverage probability approaches a value smaller than 1. This is because, given an interference ratio $\xi$, the increase of $\zeta_{A}$ not only provides the hybrid transmitter with more harvested energy to transmit, but also leads to more interference that harms the transmission. 

Fig.~\ref{CP_density_comparison}  compares coverage probability (as a function of density $\zeta_A$) of PTP, STP, pure ambient backscattering, and pure HTT. 
When $\xi$ is small (i.e., $\xi=0.2$) as shown in Fig.~\ref{fig:CP_density_comparison1}, the pure HTT transmitter achieves significantly  higher coverage probabilities. However, in the case with high interference ratio (i.e., $\xi=0.8$) as depicted in Fig.~\ref{fig:CP_density_comparison2}, their performance gap becomes smaller and pure ambient backscattering outperforms pure HTT when $\zeta_{A}$ is large (e.g., above 0.06). This is because with this condition, $\mathcal{C}_{\mathrm{H}}$ is adversely affected by the high interference. We also observe that PTP achieves similar performance to that of STP under small $\zeta_{A}$ and is obviously outperformed by STP as $\zeta_{A}$ grows larger (e.g., above 0.06). The reason behind this is that PTP selects operation mode solely based on the incident power and is unaware of the interference level. Consequently, it remains in the HTT mode even when the achieved SINR is low. This reflects that STP is more suitable for the use in an interference rich environment.

\section{Conclusion} 
 
In this paper, we have introduced a novel paradigm of hybrid D2D communications that integrate ambient backscattering and wireless-powered communications. To enable the operation of our proposed hybrid transmitter in diverse environments, two mode selection protocols, namely PTP and STP, have been devised based on the energy harvesting rate and received SNR of the modulated backscatter, respectively. Through repulsive point process modeling, we have characterized the coverage probability of the hybrid D2D communications. The performance analysis has  shown that the self-sustainable D2D communications benefit from larger repulsion, transmission load and density of the ambient transmitters. Moreover, we have found that PTP is more suitable for the use in the scenarios with a large density of ambient transmitters and low interference level. On the contrary, STP becomes favorable in the scenarios when the interference level and density of ambient transmitters are both low or both high. 
\section{Acknowledgement} 
\vspace{-1mm}
This work was supported by the Natural Sciences and Engineering Research Council of Canada, National Research Foundation of Korea Grant Funded by the Korean Government (MSIP) under Grant 2014R1A5A1011478, and Singapore MOE Tier 1 (RG33/16 and RG18/13) and MOE Tier 2 (MOE2014-T2-2-015 ARC4/15 and MOE2013-T2-2-070 ARC16/14).   


\end{document}